\documentclass[conference,a4paper]{IEEEtran}

\usepackage[utf8]{inputenc}
\usepackage[T1]{fontenc}
\usepackage{url}
\usepackage{ifthen}
\usepackage{cite}
\usepackage[cmex10]{amsmath} 

\usepackage{bm}
\usepackage{amssymb, amsthm}
\usepackage{amsopn}
\usepackage{graphicx}
\usepackage[vlined,ruled]{algorithm2e}
\usepackage{algpseudocode}

\usepackage{tikz}
\usetikzlibrary{positioning,arrows,fit,matrix,external} 
\usepackage{pgfplots}
\pgfplotsset{compat=newest}

\usepackage{comment}

\newtheorem{thm}{Theorem}
\newtheorem{lem}{Lemma}

\newcommand{\E}{\operatorname{E}}

\newcommand{\wH}{\text{w}_{\text{H}}}
\newcommand{\Ber}{\operatorname{Ber}}
\newcommand{\Bin}{\operatorname{Bin}}

\newcommand{\rank}{\operatorname{rank}}

\newcommand{\N}{\mathbb{N}}

\newcommand{\R}{\mathbb{R}}

\newcommand{\T}{\mathsf{T}}

\newcommand{\calA}{\mathcal{A}}

\newcommand{\calE}{\mathcal{E}}

\newcommand{\calU}{\mathcal{U}}
\newcommand{\calX}{\mathcal{X}}

\newcommand{\vect}[1]{\bm{#1}}                          
\newcommand{\ip}[2]{\left\langle #1, #2 \right\rangle}  

\begin{document}
\title{Online Memorization of Random Firing Sequences \\ by a Recurrent Neural Network}


\author{%
  \IEEEauthorblockN{Patrick Murer and Hans-Andrea Loeliger}
  \IEEEauthorblockA{Dept. of Information Technology and Electrical Engineering (D-ITET) \\
                    ETH Z\"urich \\
                    Email: \{murer, loeliger\}@isi.ee.ethz.ch}
}


\maketitle

\begin{abstract}
This paper studies the capability of a recurrent neural network model
to memorize random dynamical firing patterns by a simple local learning rule.
Two modes of learning/memorization are considered:
The first mode is strictly online, with a single pass through the data, 
while the second mode uses multiple passes through the data. 
In both modes, the learning is strictly local (quasi-Hebbian): At any given time step, 
only the weights between the neurons firing (or supposed to be firing) at the previous time step 
and those firing (or supposed to be firing) at the present time step are modified.

The main result of the paper is an upper bound on the probability 
that the single-pass memorization is not perfect. 
It follows that the memorization capacity in this mode asymptotically 
scales like that of the classical Hopfield model 
(which, in contrast, memorizes static patterns). 
However, multiple-rounds memorization is shown to achieve 
a higher capacity (with a nonvanishing number of bits per connection/synapse).
These mathematical findings 
may be helpful for understanding the functions of
short-term memory and long-term memory in neuroscience.
\end{abstract}


\section{Introduction}
\label{sec:Intro}

In this paper, we study the capability 
of a simple recurrent neural network to memorize 
and to reproduce a random dynamical firing pattern. 

The background of this paper are neural networks with spiking neurons
\cite{VGT:spms2005} -- \cite{MarziHespanhaMadhow2018}.
Such networks may be studied 
either as models of biological neural networks, 
or as candidates for neuromorphic hardware, 
or as a mode of mathematical signal processing as in \cite{LoeligerNeff2015}. 
In any case, memorizing long sequences of firing patterns 
must be an elementary capability of such networks. 
The rare phenomenon of a photographic memory may here remind us
of the feats of memorization routinely performed in everyday activities.

The classic reference for memorization is the Hopfield network
\cite{Hopfield1982}, \cite[Chapter 42]{MacKayBook}.
Recurrent networks with higher capacities have been proposed in
\cite{KarbasiEtAl2013} -- \cite{ChaudhuriFiete2019}.
However, all these networks memorize static vectors (as static attractors 
of a dynamical network).
By contrast, in this paper, we study the memorization of 
dynamical firing sequences,
which seems to have been somewhat neglected in the literature. 

The present paper is not immediately related to the vast literature 
on (nonspiking) recurrent neural networks such as LSTM networks
\cite{LSTM1997} and others \cite{Gardner1988} -- \cite{MiyoshiEtAl2003}.

We will consider two different modes of learning. 
The first mode is strictly online, with a single pass through the data;
the second mode uses multiple passes through the data. 
In both modes, the learning is strictly local, or quasi Hebbian: 
At any given time $n$, only the weights between the neurons firing 
(or supposed to be firing) at time $n-1$ 
and those firing (or supposed to be firing) at time $n$ are modified.
The first mode may thus be viewed as a model for instantaneous learning 
in short-term memory.

The main result of this paper is an upper bound on the probability 
that the single-pass memorization is not perfect. 
From this bound, it follows that the asymptotic memorization capacity 
in the strict online mode is at least $O\big(L/\ln(L)\big)$ bits per neuron,
which vanishes in terms of bits per connection (i.e., per synapse).
By contrast, multiple-rounds memorization is easily seen to 
achieve a significantly higher capacity, 
with a nonvanishing number of bits per connection/synapse. 
The (important) ability of single-pass online memorization
thus appears to be bought at the expense of a smaller capacity,
which may be of interest for understanding the functions 
of short-term memory and long-term memory in neuroscience
\cite{Miller1956} -- \cite{Cowan2008}.

The paper is structured as follows. The network model is defined in Section~\ref{sec:NetworkModel}.
Section~\ref{sec:LearningRules} introduces the considered learning rules. 
The main result---an upper bound on the probability of imperfect
single-pass memorization---is stated in Section~\ref{sec:Theorem}.
The bulk of the paper is Section~\ref{sec:ProofThm}, which proves the bound of Section~\ref{sec:Theorem}.
Section~\ref{sec:MultiPassMem} investigates multi-pass memorization via a least-squares approach.
The asymptotic memorization capacity of both learning modes is addressed in Section~\ref{sec:Capacity},
and Section~\ref{sec:Conclusion} concludes the paper.


\section{The Network Model}
\label{sec:NetworkModel}

We consider a discrete-time network model with $L$ neurons
$\xi_1,\ldots,\xi_L$ as follows.
Each neuron is a map $\xi_\ell \colon \R^L \rightarrow \{0,1\}$
defined as
\begin{align}
\vect{y} \mapsto \xi_\ell(\vect{y}) := \begin{cases}
                                            1, &\text{if } \ip{\vect{y}}{\vect{w}_\ell} + \eta_\ell \geq \theta_\ell \\
                                            0, &\text{otherwise,}
                                       \end{cases}
\end{align}
which is characterized by a weight vector $\vect{w}_\ell \in \R^L$
and a threshold $\theta_\ell\in\R$ and where 
$\ip{\vect{y}}{\vect{w}_\ell} := \vect{w}_\ell^\T\vect{y}$ is the standard inner product.
The quantity $\eta_\ell$ is an arbitrary bounded disturbance (or error) with
\begin{align}
    -\eta \leq \eta_\ell \leq \eta, \label{eq:DefBoundedNoise1}
\end{align}
which subsumes imprecise computations and freak firings.
In our main result, $\eta$ will be allowed to grow
linearly with $L$, cf.\ (\ref{eq:DefThreshold}) and (\ref{eq:LConstraint}) below.

These neurons are connected to form an autonomous
recurrent network producing the signal (firing sequence) 
$\vect{y}[1], \vect{y}[2], \ldots \in \{0,1\}^L$ with 
\begin{align} \label{eqn:NetworkRecurrence}
    \vect{y}[k+1] := \big( \xi_1(\vect{y}[k]),\ldots, \xi_L(\vect{y}[k]) \big)^\T
\end{align}
beginning from some initial value $\vect{y}[0] \in\R^L$.

In this paper, we want the network to reproduce a signal 
(i.e., a firing sequence) of length $N\geq 2$ that is given in the form of a matrix 
$\bm{A} = (\vect{a}_1,\ldots,\vect{a}_N) \in\{0,1\}^{L \times N}$ 
with columns $\vect{a}_1,\ldots,\vect{a}_N \in\{0,1\}^L$,
i.e., we want (\ref{eqn:NetworkRecurrence}),
when initialized with 
\begin{align} 
    \vect{y}[0]=\vect{a}_0 := \vect{a}_N
\end{align}
to yield
\begin{align} 
    \vect{y}[k] = \vect{a}_{(k \bmod N)} \label{eqn:RepeatSignalA}
\end{align}
for $k=1,2,\ldots$, repeating the columns of $\vect{A}$  forever.

Such a network can be used as an associative memory as follows:
When initialized with an arbitrary column of $\bm{A}$
\begin{align} 
    \vect{y}[0] =  \vect{a}_n, \label{eqn:InitToColumn}
\end{align}
the network will produce the sequence
\begin{align}
    \vect{y}[k] = \vect{a}_{((k+n) \bmod N)},\quad k=1, 2, \ldots
\end{align}


\section{Learning Rules}
\label{sec:LearningRules}

Given the matrix $\bm{A} = (a_{\ell,n})$ (where $a_{\ell,n}$ is the entry in row $\ell$ and column $n$), 
we consider learning rules of the following form.
Starting from some initial value $\vect{w}_\ell^{(0)}\in\R^L$ the weights are updated recursively by
\begin{align} 
    \vect{w}_\ell^{(n)} = \vect{w}_\ell^{(n-1)} + \Delta\vect{w}_{\ell,n}, \quad n=1,\ldots,K, \label{eq:OnlineMemRule}
\end{align}
where the weight increment $\Delta\vect{w}_{\ell,n}$
of neuron $\xi_\ell$ at time $n$ depends only on $a_{\ell,n}$ 
(the desired behavior of this neuron at this time)
and on the preceding firing vector $\vect{a}_{n-1}$,
and perhaps also on the previous weights $\vect{w}_\ell^{(n-1)}$
of this neuron.

This mode of learning may be called quasi-Hebbian
since the stated restrictions on $\Delta\vect{w}_{\ell,n}$
essentially agree with those of Hebbian learning \cite{Hebb1949},
except that the term ``Hebbian'' is normally reserved for unsupervised learning.
The point of these restrictions is their suitability for hardware implementation,
both biological and neuromorphic.

We will consider two versions of (\ref{eq:OnlineMemRule}).
In the first version (cf.\ Section~\ref{sec:Theorem}),
we pass through the data exactly once, i.e., $K=N$, and
\begin{align}
    \Delta\vect{w}_{\ell,n} := a_{\ell,n}\big(\vect{a}_{n-1}-p\vect{1}_L\big), \label{eq:WeightIncSingle}
\end{align}
where
\begin{align}
    \vect{1}_L := \big(1,1,\ldots,1\big)^\T \in \R^L,
\end{align}
and $0<p<1$ is defined in Section~\ref{sec:Theorem}.
In the second version (in Section~\ref{sec:MultiPassMem}),
we allow multiple passes through the data, i.e., $K\gg N$, and
\begin{align}
    \Delta\vect{w}_{\ell,n} := \beta^{(n)}\Big(a_{\ell,n} - \big\langle\vect{a}_{n-1},\vect{w}_\ell^{(n-1)}\big\rangle\Big)\vect{a}_{n-1}, \label{eq:WeightIncMulti}
\end{align}
for some step size $\beta^{(n)}>0$.


\section{Single-pass Memorization -- Main Result}
\label{sec:Theorem}

For a network as in Section~\ref{sec:NetworkModel},
we now analyze the probability of perfect memorization
for a random matrix $\bm{A}\in\{0,1\}^{L\times N}$ with i.i.d.
entries $a_{\ell,n}$ parameterized by
\begin{align}
    p := \Pr[a_{\ell,n}=1],
\end{align}
which we denote by $\bm{A} \overset{\text{i.i.d.}}{\sim} \Ber(p)^{L\times N}$.

The weight vectors are defined as
\begin{align}
    \vect{w}_{\ell} := \vect{w}_{\ell}^{(N)} \label{eq:OnlineMemRule_Tot}
\end{align}
where $\vect{w}_{\ell}^{(N)}$ is defined recursively as
\begin{align}
    \vect{w}_{\ell}^{(n)} := \begin{cases}
                                    \vect{w}_{\ell}^{(n-1)},                                &\text{if } a_{\ell,n}=0 \\
                                    \vect{w}_{\ell}^{(n-1)} + \vect{a}_{n-1}-p\vect{1}_L,   &\text{if } a_{\ell,n}=1,
                              \end{cases} \label{eq:SinglePassDeltaWeight}
\end{align}
and $\vect{w}_{\ell}^{(0)}=\vect{0}$, for $n=1,\ldots,N$,
as in (\ref{eq:WeightIncSingle}), resulting in
\begin{align} 
    \vect{w}_\ell = \sum_{j\in J_\ell} \big(\vect{a}_{j-1} - p\vect{1}_L\big)
                  = \sum_{j\in J_\ell} \vect{a}_{j-1} - |J_\ell|p\vect{1}_L, \label{eq:DefWeights}
\end{align}
where $J_\ell$ is the set 
\begin{align}
    J_\ell := \{ n\in \{ 1,\ldots, N \}: a_{\ell,n}=1 \}
\end{align}
of desired firing positions of neuron $\xi_\ell$ and $|J_\ell|$ denotes its cardinality.
It is easily verified that
\begin{align}
    \E[\vect{w}_\ell] = \vect{0}.
\end{align}

Let $\calE_{\bm{A}}$ be the event that the memorization of $\bm{A}$ is not perfect.
Our main result is the following theorem.

\begin{thm}[\bfseries{Upper Bound on $\Pr[\calE_{\bm{A}}]$}] \label{thm:MainThm}
For all integers $L\geq 1$ and $N\geq 2$, $0<p<1$, $\bm{A}\overset{\text{i.i.d.}}{\sim} \Ber(p)^{L\times N}$,
the recurrent network with weight vectors (\ref{eq:DefWeights}), thresholds
\begin{align}
    \theta_\ell := \theta := \frac{1}{4}Lp(1-p), \quad \ell=1,\ldots,L, \label{eq:DefThreshold}
\end{align}
disturbance bound
\begin{align}
    \eta := \tilde\eta \cdot \theta, \quad 0<\tilde\eta<1, \label{eq:LConstraint}
\end{align}
and initialized with any column of $\bm{A}$
will reproduce a periodic extension of $\bm{A}$ with
\begin{align}
    &\Pr[\calE_{\bm{A}}] \nonumber \\
    &\quad< 2LN e^{-\frac{1}{8}(1-\tilde\eta)^2p^2(1-p)^2\frac{L}{N}} + LN e^{-D_{\text{KL}}\left(\left.\!\frac{1+\tilde\eta}{2}p \right\| p\right) L}, \label{eq:ThmBound}
\end{align}
where $D_{\text{KL}}(p_1\!\!\!\parallel\!\!p_2)$ denotes the
Kullback--Leibler divergence (as defined in (\ref{eq:DefKLDivergence}) below)
between two Bernoulli distributions with success probabilities $0<p_1,p_2<1$.
\hfill $\Box$
\end{thm}

In consequence, a sufficient condition for the bound in (\ref{eq:ThmBound})
to vanish for $L\to\infty$ is
\begin{align}
    N \leq \frac{1}{8}(1-\tilde\eta)^2p^2(1-p)^2\frac{L}{2\ln(L)}. \label{eq:Nconstraint}
\end{align}
Some numerical examples are given in Figure~\ref{fig:LNCurvesWithConstErrProb},
which plots $L$ vs. $N$ for the right-hand side of (\ref{eq:ThmBound})
to achieve some desired level.

Clearly, for all $\varepsilon > 0$, there exists $L_\varepsilon \in \N$
such that $L^2/\ln(L) \geq L^{2-\varepsilon}$ for all $L\geq L_\varepsilon$.
It follows that
\begin{align}
    LN\geq L^{2-\varepsilon}
\end{align}
for $N=L/\ln(L)$ and $L\to\infty$,
i.e., asymptotically the network is able to memorize almost
square matrices with instantaneous learning as in (\ref{eq:OnlineMemRule_Tot}) -- (\ref{eq:DefWeights}).


\section{Proof of Theorem~\ref{thm:MainThm}}
\label{sec:ProofThm}

We now prove Theorem~\ref{thm:MainThm}, by using the union bound
and by upper bounding the error probability for a single entry $a_{\ell,n}$
which amounts to bound the tails of $\ip{\vect{a}_{n-1}}{\vect{w}_\ell}$.

The memorization is perfect if and only if $\xi_\ell(\bm{a}_{n-1}) = a_{\ell,n}$
for all $\ell\in\{1,\ldots,L\}$ and for all $n\in\{1,\ldots,N\}$.
By the union bound, we have
\begin{align}
    \Pr[\calE_{\bm{A}}] \leq \sum_{\ell=1}^L \sum_{n=1}^N \Pr\!\left[\xi_\ell(\bm{a}_{n-1})\neq a_{\ell,n} \right]. \label{eq:SumOfErrProb}
\end{align}
Moreover, using the same threshold $\theta$ for each neuron
and by the law of total probability, we have
\begin{align}
    \Pr\!&\left[\xi_\ell(\bm{a}_{n-1})\neq a_{\ell,n} \right] \nonumber \\
            &\qquad = (1-p)\Pr\!\left[\left.\ip{\bm{a}_{n-1}}{\bm{w}_\ell} + \eta_\ell \geq \theta \,\right| a_{\ell,n}=0\right] \nonumber \\
            &\qquad\quad + p\Pr\!\left[\left.\ip{\bm{a}_{n-1}}{\bm{w}_\ell} + \eta_\ell < \theta \,\right| a_{\ell,n}=1\right]. \label{eq:ProbSumOfTwoErrEvents}            
\end{align}

Now, let $\ell\in\{1,\ldots,L\}$ and let $n\in\{1,\ldots,N\}$
be fixed but arbitrary. Then
\begin{align}
    \ip{\vect{a}_{n-1}}{\vect{w}_\ell} &= \ip{\vect{a}_{n-1}}{\sum_{j\in J_\ell} \big(\vect{a}_{j-1} - p\vect{1}_L\big)} \\
        &= \sum_{j\in J_\ell} \big\langle \vect{a}_{n-1}, \underbrace{\vect{a}_{j-1} - \E[\vect{a}_{j-1}]}_{\qquad\,\,=:\,\tilde{\vect{a}}_{j-1}} \big\rangle \\
        &= \sum_{j=1}^N a_{\ell,j} \ip{\vect{a}_{n-1}}{\tilde{\vect{a}}_{j-1}} \\
        &= a_{\ell,n} \ip{\vect{a}_{n-1}}{\tilde{\vect{a}}_{n-1}} + S_{\ell,n}, \label{eq:CorrelationA}
\end{align}
where
\begin{align}
    S_{\ell,n} := \sum_{j\in\{1,\ldots,N\}\setminus\{n\}} a_{\ell,j} \ip{\vect{a}_{n-1}}{\tilde{\vect{a}}_{j-1}}\!. \label{eq:Def_selln}
\end{align}

\begin{lem} \label{lem:Sigma}
The random variable $S_{\ell,n}$ as defined in (\ref{eq:Def_selln}) has expectation zero, i.e.,
\begin{align}
    \E[S_{\ell,n}] &= 0, \label{eq:lemmaS_A}
\end{align}
and its moment generating function is upper bounded by
\begin{align}
    \E\!\left[e^{tS_{\ell,n}}\right] < e^{\frac{t^2}{8}LN} \label{eq:lemmaS_B}
\end{align}
for all $t\in\R$. \hfill $\Box$
\end{lem}
\noindent
A proof of Lemma~\ref{lem:Sigma} is given in Appendix~\ref{appendixLemmaSigma}.

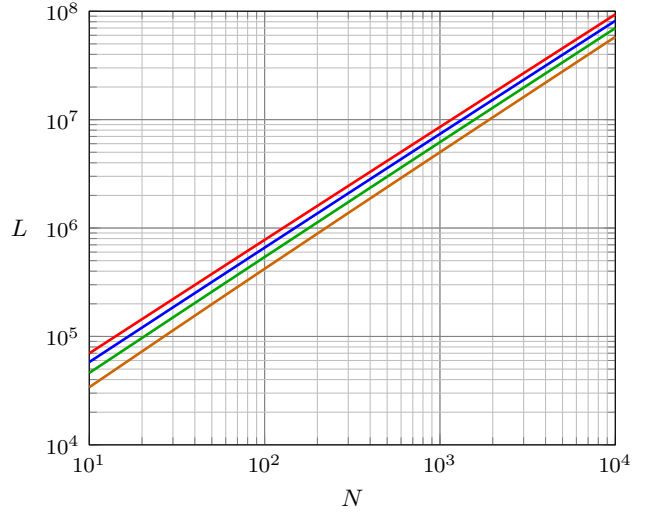
\begin{figure}
\centering
\definecolor{mycolor1}{rgb}{0,0.65,0}
\definecolor{mycolor2}{rgb}{0.8, 0.4, 0}


\begin{tikzpicture}
    \begin{loglogaxis}[
        ylabel near ticks,
        ylabel style={rotate=-90},
        small,
        width=8.5cm,
        xmin = 10,
        xmax = 10^4,
        ymin = 10^4,
        ymax = 10^8,
        xlabel=$N$,
        ylabel=$L$,
        grid=both,
        major grid style=gray,
        minor grid style=lightgray
    ]

    \addplot[color=mycolor2, line width=1pt] coordinates{
    (10,34002)
    (30,113563)
    (60,241653)
    (100,420558)
    (300,1376234)
    (600,2896648)
    (1000,5004563)
    (1300,6623881)
    (2000,10488293)
    (2300,12172543)
    (3000,16152375)
    (4000,21933480)
    (5000,27801557)
    (6000,33738916)
    (7000,39733872)
    (8000,45778122)
    (10000,57991072)
    };
    \addplot[color=mycolor1, line width=1pt] coordinates{
    (10,46058)
    (30,149591)
    (60,313558)
    (100,540231)
    (300,1734291)
    (600,3611697)
    (1000,6195106)
    (1300,8170821)
    (2000,12866351)
    (2300,14906644)
    (3000,19716982)
    (4000,26684042)
    (5000,33737646)
    (6000,40860198)
    (7000,48040076)
    (8000,55269021)
    (10000,69850794)
    };
    \addplot[color=blue, line width=1pt] coordinates{
    (10,57992)
    (30,185311)
    (60,384906)
    (100,659041)
    (300,2090110)
    (600,4322640)
    (1000,7379211)
    (1300,9709646)
    (2000,15232526)
    (2300,17627293)
    (3000,23264556)
    (4000,31412597)
    (5000,39646881)
    (6000,47949878)
    (7000,56310008)
    (8000,64719042)
    (10000,81660591)
    };
    \addplot[color=red, line width=1pt] coordinates{
    (10,69851)
    (30,220836)
    (60,455896)
    (100,777286)
    (300,2444423)
    (600,5030776)
    (1000,8558869)
    (1300,11242833)
    (2000,17590376)
    (2300,20338492)
    (3000,26800105)
    (4000,36125536)
    (5000,45536985)
    (6000,55016973)
    (7000,64553950)
    (8000,74139710)
    (10000,93434406)
    };
    \end{loglogaxis}
\end{tikzpicture}
\caption{
Value of $L$ required for the right-hand side of (\ref{eq:ThmBound})
to equal $10^{-3}$, $10^{-6}$, $10^{-9}$, $10^{-12}$ (from bottom to top)
for $p = 1/2$, and $\tilde\eta = 1/8$.
}
\label{fig:LNCurvesWithConstErrProb}
\end{figure}

Let us define the event
\begin{align}
    \calE_{a_{\ell,n}} := \big\{ \xi_\ell(\bm{a}_{n-1}) \neq a_{\ell,n} \big\}.
\end{align}
Then by (\ref{eq:CorrelationA}) we can upper bound (\ref{eq:ProbSumOfTwoErrEvents}) as
\begin{align}
    \Pr\!&\left[\calE_{a_{\ell,n}}\right] \nonumber \\
            &\,\,\,\leq \Pr\!\left[\left.S_{\ell,n} \geq \theta - \eta_\ell \,\right| a_{\ell,n}=0\right] \nonumber \\
            &\,\,\,\quad + \Pr\!\left[\left.\ip{\vect{a}_{n-1}}{\tilde{\vect{a}}_{n-1}} + S_{\ell,n} < \theta - \eta_\ell \,\right| a_{\ell,n}=1\right]. \label{eq:UppBound_Eln}        
\end{align}
As (global) threshold we choose
\begin{align}
    \theta := \frac{1}{4} \sum_{a\in\{0,1\}} \E\!\left[\left.\ip{\vect{a}_{n-1}}{\vect{w}_\ell} \,\right| a_{\ell,n}=a\right],
\end{align}
cf.\ Figure~\ref{fig:SketchProbDist}, and it can be shown that 
\begin{align}
    \theta = \frac{1}{4} \E\!\left[\ip{\vect{a}_{n-1}}{\tilde{\vect{a}}_{n-1}}\right]. \label{eq:ThresholdDef_B}
\end{align}
Note that $S_{\ell,n}$ depends on $a_{\ell,n}$.
To get rid of the conditioning on $a_{\ell,n}$ in (\ref{eq:UppBound_Eln}), 
we observe that an error, i.e., $\calE_{a_{\ell,n}}$ implies either
$|S_{\ell,n}| \geq \theta - \eta_\ell$, or $|S_{\ell,n}| < \theta - \eta_\ell$ and
$\ip{\vect{a}_{n-1}}{\tilde{\vect{a}}_{n-1}} + S_{\ell,n} < \theta - \eta_\ell$, cf.\ Figure~\ref{fig:SketchProbDist}.
Thus by the union bound, we obtain
\begin{align}
    \Pr\!\left[\calE_{a_{\ell,n}}\right]
            &\leq \Pr\!\left[|S_{\ell,n}| \geq \theta-\eta_\ell \right] \nonumber \\
            &\qquad+ \Pr\!\left[\ip{\vect{a}_{n-1}}{\tilde{\vect{a}}_{n-1}} < 2(\theta-\eta_\ell) \right] \\
            &\leq \Pr\!\left[|S_{\ell,n}| \geq \theta-\eta \right] \nonumber \\ 
            &\qquad+ \Pr\!\left[\ip{\vect{a}_{n-1}}{\tilde{\vect{a}}_{n-1}} < 2(\theta+\eta) \right] \label{eq:UppBound_error_aln_A} \\
            &= \Pr\!\left[|S_{\ell,n}| \geq \theta(1-\tilde\eta) \right] \nonumber \\ 
            &\qquad+ \Pr\!\left[\ip{\vect{a}_{n-1}}{\tilde{\vect{a}}_{n-1}} < 2\theta(1+\tilde\eta) \right], \label{eq:UppBound_error_aln_B}
\end{align}
where in (\ref{eq:UppBound_error_aln_A}) we applied (\ref{eq:DefBoundedNoise1}),
and (\ref{eq:UppBound_error_aln_B}) holds because of (\ref{eq:LConstraint}).

Now, we apply the Chernoff bound \cite{Chernoff1952} to both terms on the right-hand side of (\ref{eq:UppBound_error_aln_B}).
Thus, we have
\begin{align}
    \Pr\!\left[S_{\ell,n} \geq \theta(1-\tilde\eta)\right]
            &\leq \min_{t>0} \frac{\E\!\left[e^{tS_{\ell,n}}\right]}{e^{t\theta(1-\tilde\eta)}} \label{eq:ChernS_elln_A} \\
            &< \min_{t>0} \frac{e^{\frac{t^2}{8}LN}}{e^{t\theta(1-\tilde\eta)}} \label{eq:ChernS_elln_B} \\
            &= e^{-\frac{2\theta^2(1-\tilde\eta)^2}{LN}}. \label{eq:ChernS_elln_C}
\end{align}
The step from (\ref{eq:ChernS_elln_A}) to (\ref{eq:ChernS_elln_B}) follows from (\ref{eq:lemmaS_B}).
The bound (\ref{eq:ChernS_elln_B}) is minimized by $t_{\min} = 4\theta(1-\tilde\eta)/(LN)$
which implies (\ref{eq:ChernS_elln_C}).

The lower tail of $S_{\ell,n}$, i.e., $\Pr\!\left[S_{\ell,n} \leq -(\theta-\eta)\right]$
can be upper bounded analogously.
Thus by the union bound of both tails, we obtain
\begin{align}
    \Pr\!\left[|S_{\ell,n}| \geq \theta(1-\tilde\eta) \right] < 2e^{-\frac{2\theta^2(1-\tilde\eta)^2}{LN}}. \label{eq:UpperBound_S}
\end{align}

As for the other term on the right-hand side of (\ref{eq:UppBound_error_aln_B}),
we note
\begin{align}
    \ip{\vect{a}_{n-1}}{\tilde{\vect{a}}_{n-1}}
        &= \sum_{\ell=1}^L a_{\ell,n-1}(a_{\ell,n-1}-p) \\
        &= (1-p)\sum_{\ell=1}^L a_{\ell,n-1}
\end{align}
since $a_{1,n-1},\ldots,a_{L,n-1} \,\overset{\text{i.i.d.}}{\sim}\, \Ber(p)$, thus
\begin{align}
    \frac{1}{1-p} \ip{\vect{a}_{n-1}}{\tilde{\vect{a}}_{n-1}} \,\sim\, \Bin(L,p), \label{eq:InnerProd_Bin}
\end{align}
which together with (\ref{eq:ThresholdDef_B}) implies (cf.\ (\ref{eq:DefThreshold}))
\begin{align}
    \theta = \frac{1}{4}Lp(1-p). \label{eq:ThresholdDef_A}
\end{align}
Then, inserting (\ref{eq:ThresholdDef_A}) into the right summand on the
right-hand side of (\ref{eq:UppBound_error_aln_B}) yields
\begin{align}
    &\Pr\!\left[\ip{\vect{a}_{n-1}}{\tilde{\vect{a}}_{n-1}} < \frac{1+\tilde\eta}{2}Lp(1-p) \right] \nonumber \\
        &\qquad\qquad = \Pr\!\left[\frac{1}{1-p}\ip{\vect{a}_{n-1}}{\tilde{\vect{a}}_{n-1}} < \frac{1+\tilde\eta}{2}Lp \right] \label{eq:ChernIP_B} \\
        &\qquad\qquad \leq e^{-D_{\text{KL}}\left(\left.\!\frac{1+\tilde\eta}{2}p \right\| p\right) L}, \label{eq:ChernIP_C}
\end{align}
with Kullback--Leibler divergence (or relative entropy)
\begin{align}
    D_{\text{KL}}(p_1\!\!\parallel\!p_2) := p_1 \ln\!\left(\frac{p_1}{p_2}\right) + (1-p_1) \ln\!\left(\frac{1-p_1}{1-p_2}\right)\!, \label{eq:DefKLDivergence}
\end{align}
for $0<p_1,p_2<1$, cf.\ \cite{CoverThomasBook}.
From (\ref{eq:ChernIP_B}) to (\ref{eq:ChernIP_C}) we applied Lemma~\ref{lem:ChBoundBinomial}
(which is stated in Appendix~\ref{appendixChBoundBinomial})
with $1-\delta = (1+\tilde\eta)/2$, $0<\tilde\eta<1$, because of (\ref{eq:InnerProd_Bin}).
Note that in general
\begin{align}
    D_{\text{KL}}(p_1\!\!\parallel\!p_2) \neq D_{\text{KL}}(p_2\!\!\parallel\!p_1),
\end{align}
and for all $0<p_1,p_2<1$
\begin{align}
    D_{\text{KL}}(p_1\!\!\parallel\!p_2) \geq 0
\end{align}
with equality if and only if $p_1=p_2$.

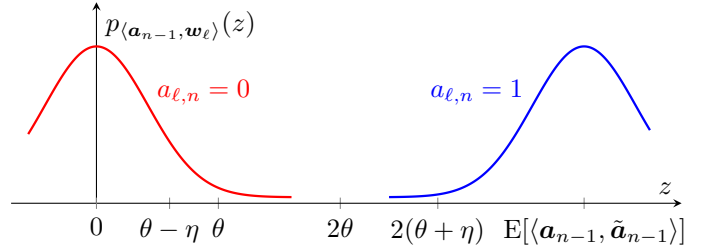
\begin{figure}
\centering
\centering
\begin{tikzpicture}
\begin{axis}[
	width=10.4cm, height=4.25cm,
    ymax = 0.5, ymin=-0.0005,
    ytick=\empty,
    xmax = 12, xmin=-1.75,
    xtick={0, 1.5, 2.5, 5, 7, 10},
    xticklabels={$0$, $\theta-\eta$, $\theta$, $2\theta$, $2(\theta+\eta)$, $\,\,\,\,\,\E[ \ip{\vect{a}_{n-1}}{\tilde{\vect{a}}_{n-1}} ]$},
    extra x ticks=0,
    extra x tick labels=$0$,
	axis x line=center,
	axis y line=center,
	xlabel=$z$,
    ylabel=$p_{\ip{\vect{a}_{n-1}}{\vect{w}_\ell}}(z)$
]

\addplot[
    line width=0.03cm,
	red,
	domain=-1.4:4.0,
	samples=200,
]{0.015 + 0.375*exp(-x^2/2.25)};

\node[] at (axis cs: 2.2,0.275) {$\textcolor{red}{a_{\ell,n}=0}$};

\addplot[
    line width=0.03cm,
	blue,
	domain=6.0:11.35,
	samples=200,
]{0.015 + 0.375*exp(-(x-10)^2/2.25)};

\node[] at (axis cs: 7.8,0.275) {$\textcolor{blue}{a_{\ell,n}=1}$};

\end{axis}
\end{tikzpicture}
\caption{
Sketch of the probability distribution of (\ref{eq:CorrelationA}) for the realization
$\ip{\vect{a}_{n-1}}{\tilde{\vect{a}}_{n-1}} = \E[\ip{\vect{a}_{n-1}}{\tilde{\vect{a}}_{n-1}}]$
and the two cases $a_{\ell,n}=0$ (peak on the left) and $a_{\ell,n}=1$ (peak on the right).
}
\label{fig:SketchProbDist}
\end{figure}

Finally, we obtain
\begin{align}
    \Pr\!\left[\calE_{a_{\ell,n}}\right]
        &< 2e^{-\frac{2\theta^2(1-\tilde\eta)^2}{LN}} + e^{-D_{\text{KL}}\left(\left.\!\frac{1+\tilde\eta}{2}p \right\| p\right) L} \label{eq:UpperBoundSingleNeuron_A} \\
        &= 2e^{-\frac{1}{8}(1-\tilde\eta)^2p^2(1-p)^2\frac{L}{N}} + e^{-D_{\text{KL}}\left(\left.\!\frac{1+\tilde\eta}{2}p \right\| p\right) L}. \label{eq:UpperBoundSingleNeuron_B}
\end{align}
Inequality (\ref{eq:UpperBoundSingleNeuron_A}) follows from (\ref{eq:UppBound_error_aln_B})
together with the two upper bounds (\ref{eq:UpperBound_S}) and (\ref{eq:ChernIP_C}).
In (\ref{eq:UpperBoundSingleNeuron_B}) we inserted (\ref{eq:ThresholdDef_A}).

The upper bound in (\ref{eq:UpperBoundSingleNeuron_B}) is independent on $\ell$ and $n$,
and thus (\ref{eq:SumOfErrProb}) yields (\ref{eq:ThmBound})
which concludes the proof. \hfill $\blacksquare$


\section{Multi-pass Memorization}
\label{sec:MultiPassMem}

Perfect memorization can also be achieved via a certain
least-squares problem, and solving this least-squares problem
via stochastic gradient descent can be phrased as
multi-pass learning according to (\ref{eq:WeightIncMulti}).

Specifically, for fixed $\ell\in\{1,\ldots,L\}$,
consider the least-squares problem
\begin{align}
    \min_{\vect{w}_\ell} \sum_{n=1}^N \left|\ip{\vect{a}_{n-1}}{\vect{w}_\ell} - a_{\ell,n}\right|^2\
        = \min_{\vect{w}_\ell} \big\lVert \tilde{\bm{A}}\vect{w}_\ell - \tilde{\vect{a}}_\ell \big\rVert^2, \label{eq:compLS}   
\end{align}
where
\begin{align}
    \tilde{\bm{A}} := \begin{pmatrix} 
                            \bm{a}_{N}^\T \\
                            \bm{a}_{1}^\T \\
                            \vdots \\                    
                            \bm{a}_{N-1}^\T
                      \end{pmatrix} \in \R^{N\times L},
    \quad \tilde{\bm{a}}_\ell := \begin{pmatrix} 
                                    a_{\ell,1} \\
                                    \vdots \\
                                    a_{\ell,N}
                                  \end{pmatrix} \in \R^N. \label{eq:defTildeA}
\end{align}
Note that $\tilde{\bm{A}}$ is the transposed matrix of
$(\bm{a}_N,\bm{a}_1,\ldots,\bm{a}_{N-1})\in\R^{L\times N}$, i.e.,
of the one time-step cyclic shifted version of $\bm{A}$, 
and $\tilde{\bm{a}}_\ell$ is the $\ell$-th row of $\bm{A}$
turned into a column vector.

If $\rank(\tilde{\bm{A}}) = N$, then
\begin{align}
    \min_{\vect{w}_\ell\in\R^L} \big\lVert \tilde{\bm{A}}\vect{w}_\ell - \tilde{\vect{a}}_\ell \big\rVert^2 = 0, \label{eq:LS_zeroError} 
\end{align}
which implies that $\bm{A}$ is (perfectly) memorizable, i.e., $\Pr[\calE_{\bm{A}}]=0$.
For $L\geq N$, $0<p\leq 1/2$, $\bm{A}\overset{\text{i.i.d.}}{\sim}\Ber(p)^{L\times N}$,
it follows from \cite{Tikhomirov2019} that
\begin{align}
    \Pr\!\left[\rank(\tilde{\bm{A}}) = N\right] \geq 1 - \big(1-p+o_N(1)\big)^N,
\end{align}
where $o_N(1)$ denotes a sequence which converges to zero,
i.e., $\lim_{N\to\infty} o_N(1) = 0$.
Thus, any matrix $\bm{A}\overset{\text{i.i.d.}}{\sim}\Ber(p)^{L\times N}$
with $L\geq N$, and in particular with
\begin{align}
    L=N \label{eq:IterativeNconstraint}
\end{align}
is memorizable as $N\to\infty$.

Clearly, the least-squares problem (\ref{eq:compLS})
could be solved by gradient descent as follows.
Starting from some initial guess $\vect{w}_\ell^{(0)}$ we proceed by
\begin{align}
    \vect{w}_\ell^{(n)} &= \vect{w}_\ell^{(n-1)} + \beta^{(n)}\tilde{\bm{A}}^\T \big(\tilde{\vect{a}}_\ell - \tilde{\bm{A}}\vect{w}_\ell^{(n-1)} \big) \label{eq:GradientDes_A} 
\end{align}
for $n=1,\ldots,K$, $K\in\N$, and with step size $\beta^{(n)}>0$.
The recursion (\ref{eq:GradientDes_A}) with constant $\beta^{(n)} = \beta$
converges to a minimizer of (\ref{eq:compLS}) if
\begin{align}
    0 < \beta < \frac{2}{\lambda_{\max}(\tilde{\bm{A}}^\T\tilde{\bm{A}})},
\end{align}
where $\lambda_{\max}(\tilde{\bm{A}}^\T\tilde{\bm{A}}) > 0$ is the largest
eigenvalue of $\tilde{\bm{A}}^\T\tilde{\bm{A}}$.

Finally, replacing gradient descent as in (\ref{eq:GradientDes_A})
by stochastic gradient descent yields
\begin{IEEEeqnarray}{rCl}
    \vect{w}_\ell^{(n)} &=& \vect{w}_\ell^{(n-1)} + \beta^{(n)}\Big(a_{\ell,n} - \big\langle\vect{a}_{n-1},\vect{w}_\ell^{(n-1)}\big\rangle \Big)\vect{a}_{n-1}, \label{eq:SGDupdateRule}
    \IEEEeqnarraynumspace
\end{IEEEeqnarray}
which is (\ref{eq:WeightIncMulti}).
Again, as in (\ref{eqn:RepeatSignalA}) the column indices
are taken modulo $N$ and $\vect{a}_0 := \vect{a}_N$.
It is shown in \cite{StrohmerVershynin2009} that if at every iteration
the column indices are chosen randomly, then (\ref{eq:SGDupdateRule}) converges
exponentially in expectation to a solution of (\ref{eq:LS_zeroError}).


\section{Memorization Capacity}
\label{sec:Capacity}

Let $\calA_{\text{typical}}$ be a typical set of matrices 
(in any standard sense of "typical sequences" \cite{CoverThomasBook}) for the
random matrix $\bm{A}\overset{\text{i.i.d.}}{\sim} \Ber(p)^{L\times N}$
and $|\calA_{\text{typical}}|$ denotes the cardinality of $\calA_{\text{typical}}$.
Then, we have
\begin{align}
    \lim_{L\to\infty} \frac{1}{L} \log_2 |\calA_{\text{typical}}| = H_{\text{b}}(p)N, \label{eq:NumbTypMatrices}
\end{align}
with the binary entropy function
\begin{align}
    H_{\text{b}}(p) := -p\log_2(p) - (1-p)\log_2(1-p)
\end{align}
for $0<p<1$, cf. \cite{CoverThomasBook}.

The absolute capacity of a network is equal to the total number of bits
which can be memorized by the network, thus
\begin{align}
    C_{\text{absolute}} \geq \log_2 |\calA_{\text{typical}}| \,\, \left[\,\text{bits}\,\right]\!. \label{eq:absCapacity}
\end{align}

\subsection{Capacity per Neuron}

From (\ref{eq:NumbTypMatrices}) and (\ref{eq:absCapacity}) it follows that
the asymptotic memorization capacity in bits per neuron is lower bounded by
\begin{align}
    C_{\text{neuron}} \geq H_{\text{b}}(p)N \,\, \left[\,\text{bits per neuron}\,\right]\!.
\end{align}

For the single-pass memorization rule (\ref{eq:SinglePassDeltaWeight})
we have (\ref{eq:Nconstraint}) (which is a consequence of Theorem~\ref{thm:MainThm}),
and we thus obtain
\begin{align}
    C_{\text{single-pass}} \geq C_{p,\tilde{\eta}}\frac{L}{\ln(L)} \,\, \left[\,\text{bits per neuron}\,\right] \label{eq:CsinglePerNeuron}
\end{align}
with constant
\begin{align}
    C_{p,\tilde{\eta}} := \frac{1}{16}(1-\tilde{\eta})^2p^2(1-p)^2 H_{\text{b}}(p) > 0,
\end{align}
for $0 < p,\tilde{\eta} < 1$.
For the multi-pass memorization rule (\ref{eq:GradientDes_A}),
we have (cf.\ (\ref{eq:IterativeNconstraint}))
\begin{align}
    C_{\text{multi-pass}} \geq H_{\text{b}}(p)L \,\, \left[\,\text{bits per neuron}\,\right]\!.
\end{align}
Both memorization capacities $C_{\text{single-pass}}$ and $C_{\text{multi-pass}}$
(in bits per neuron) are unbounded in $L$.

\subsection{Capacity per Connection (Synapse)}

The capacity per connection (i.e., per nonzero weight) is
\begin{align}
    C_{\text{connection}} = \frac{C_{\text{neuron}}}{\wH(\vect{w}_\ell)},
\end{align}
where $\wH(\vect{w}_\ell)$ is the Hamming weight of the $\ell$-th weight vector.
Since for both modes of memorization $\wH(\vect{w}_\ell) = L$, for all $\ell$, we obtain
\begin{align}
    C_{\text{single-pass}} &\geq C_{p,\tilde{\eta}}\frac{1}{\ln(L)} \,\, \left[\,\text{bits per connection}\,\right]\!, \\
    C_{\text{multi-pass}} &\geq H_{\text{b}}(p) \,\, \left[\,\text{bits per connection}\,\right]\!.
\end{align}
Thus, in bits per connection the capacity $C_{\text{single-pass}}$ seems to vanish,
whereas $C_{\text{multi-pass}}$ does not vanish as $L\to\infty$.

\subsection{Comparison with the Hopfield Network}

The capacity of the Hopfield model with $L$ neurons is
$L/(2\ln(L))$ vectors with the Hebbian learning rule \cite{McElieceEtAl1987} and
$L/\sqrt{2\ln(L)}$ vectors with the Storkey learning rule \cite{StorkeyValabregue1999}.
However, for a fair comparison with the results of the present paper,
it should be noted that each vector consists of $L$ random bits, resulting
in a capacity of $L/(2\ln(L))$ bits per neuron
and $L/\sqrt{2\ln(L)}$ bits per neuron, respectively.
Thus, the capacity of the Hebbian learning rule
is on the same order as the capacity of the
single-pass memorization rule, cf.\ (\ref{eq:CsinglePerNeuron}).


\section{Conclusion}
\label{sec:Conclusion}

We have studied the capability of a ``spiking'' dynamical neural network
model to memorize random firing sequences by a form of quasi-Hebbian learning.
Our main result was an upper bound on the probability
that instantaneous memorization is not perfect.
From this bound, the instantaneous-memorization capacity
of a network with $L$ neurons is (at least) $O\big(L/\ln(L)\big)$ bits per neuron.
By contrast, iterative (i.e., multi-pass) learning is shown to achieve
a capacity of $O(L)$ bits per neuron and $O(1)$ bits per connection/synapse.
These results may be useful for understanding the functions of
short-term memory and long-term memory in neuroscience
and their potential analogs in neuromorphic hardware.



\newpage

\begin{thebibliography}{29}

	\bibitem{VGT:spms2005}
	R.\ Van Rullen, R.\ Guyonneau, and S.J.\ Thorpe,
	``Spike times make sense,''
	\textit{Trends in Neurosciences,} vol.~28, no.~1, Jan.\ 2005.

    \bibitem{ANDL:frast2018}
	J.\ Anumula et al.,
	``Feature representations for neuromorphic audio spike streams,''
	\textit{Frontiers in Neuroscience,} Feb.\ 2018.

    \bibitem{TCA:annat2013}
	J.C.\ Tapson et al.,
	``Synthesis of neural networks for spatio-temporal spike pattern recognition and processing,''
	\textit{Frontiers in Neuroscience,} Aug.\ 2013.
	

	\bibitem{Izhikevich2006}
	E.M.\ Izhikevich,
	``Polychronization: computation with spikes,''
	\textit{Neural Computation,} vol.~18, no.~2, pp.~245--282, Feb.\ 2006.

	\bibitem{MarziHespanhaMadhow2018}
	Z.\ Marzi, J.\ Hespanha, and U.\ Madhow,
	``On the information in spike timing: neural codes derived from polychronous groups,''
	\textit{2018 Information Theory and Applications Workshop (ITA),} Feb.\ 2018.

    \bibitem{LoeligerNeff2015}
	H.-A.\ Loeliger and S.\ Neff,
	``Pulse-domain signal parsing and neural computation,''
	\textit{2015 IEEE Int.\ Symp.\ on Information Theory (ISIT),} Hong Kong, China, June 2015.


	
    \bibitem{Hopfield1982}
	J.J.\ Hopfield,
	``Neural networks and physical systems with emergent collective computational abilities,''
	\textit{Proc.\ Nat.\ Academy of Sciences of the USA,} vol.~79 no.~8 pp.~2554--2558, April 1982.

    \bibitem{MacKayBook}
    D.~MacKay,
    \emph{Information Theory, Inference, and Learning Algorithms.}
    Cambridge University Press, 2003.

    

    \bibitem{KarbasiEtAl2013}
    A.~Karbasi, A.H.~Salavati, and A.~Shokrollahi,
    ``Iterative learning and denoising in convolutional neural associative memories,''
    \emph{Proc.\ 30th Int.\ Conf.\ on Machine Learning,} Atlanta, Georgia, USA, June 2013.

    \bibitem{Sa:thesis}
    A.H.\ Salavati,
    \emph{Coding Theory and Neural Associative Memories with Exponential Pattern Retrieval Capacity.}
    PhD thesis 6089 at EPFL, 2014.
    
    \bibitem{ChaudhuriFiete2019}
	R.\ Chaudhuri and I.\ Fiete,
	``Bipartite expander Hopfield networks as self-decoding high-capacity error correcting codes,''
	\textit{Adv.\ Neural Inf.\ Proc.\ Systems (NeurIPS),} Dec. 2019.



	\bibitem{LSTM1997}
	S.\ Hochreiter and J.\ Schmidhuber,
	``Long short-term memory,''
	\textit{Neural Computation,} vol.~9 no.~8 pp.~1735--1780, Dec.\ 1997.

    \bibitem{Gardner1988}
	E.\ Gardner,
	``The space of interactions in neural network models,''
	\textit{J.\ Phys.\ A: Math.\ Gen.,} vol.~21 no.~1 pp.~257--270, July 1988.

    \bibitem{RCAM1991}
	T.D.\ Chiueh and R.M.\ Goodman,
	``Recurrent correlation associative memories,''
	\textit{IEEE Trans.\ on Neural Networks,} vol.~2 no.~2 pp.~275--284, March 1991.
	
	\bibitem{JankkowskiEtAl1996}
	S.\ Jankowski, A.\ Lozowski, and J.M.\ Zurada,
	``Complex-valued multistate neural associative memory,''
	\textit{IEEE Trans.\ on Neural Networks,} vol.~7 no.~6 pp.~1491--1496, Nov. 1996.
	
    \bibitem{MiyoshiEtAl2003}
	S.\ Miyoshi, H.\ Yanai, and M.\ Okada,
	``Associative memory by recurrent neural networks with delay elements,''
	\textit{Neural Networks,} vol.~17 pp.~55--63, 2004.



    \bibitem{Miller1956}
	G.A.\ Miller,
	``The magical number seven, plus or minus two: some limits on our capacity for processing information,''
	\textit{Psychological Review,} vol.~63 no.~2 pp.~81--97, 1956.
    
    \bibitem{AtkinsonEtAl1968}
	R.C.\ Atkinson and R.M.\ Shiffrin,
	``Human memory: a proposed system and its control processes,''
	\textit{The psychology of learning and motivation,} vol.~2 pp.~89--195, 1968.
    
    \bibitem{BaddeleyHitch1974}
	A.D.\ Baddeley and G.\ Hitch,
	``Working memory,''
	\textit{The psychology of learning and motivation,} vol.~8 pp.~47--89, 1974.
    
    \bibitem{Cowan2008}
	N.\ Cowan,
	``What are the differences between long-term, short-term, and working memory?,''
	\textit{Prog. Brain Res.,} vol.~169 pp.~323--338, 2008.



    \bibitem{Hebb1949}
    D.O.\ Hebb,
    \emph{The Organization of Behavior.}
    Wiley \& Sons, 1949.



    \bibitem{Chernoff1952}
	H.\ Chernoff,
	``A measure of asymptotic efficiency for tests of a hypothesis based on the sum of observations,''
	\textit{Ann.\ Math.\ Statist.,} vol.~23 no.~4 pp.~493--507, 1952.

    \bibitem{CoverThomasBook}
    T.M.\ Cover and J.A.\ Thomas,
    \emph{Elements of Information Theory.}
    Wiley \& Sons, 2006.
    
    
    
    \bibitem{Tikhomirov2019}
	K.\ Tikhomirov,
	``Singularity of random bernoulli matrices,''
	arXiv:1812.09016v4 [math.PR], Aug. 2019

    \bibitem{StrohmerVershynin2009}
	T.\ Strohmer and R.\ Vershnin,
	``A randomized Kaczmarz algorithm with exponential convergence,''
	\textit{Journal of Fourier Analysis and Applications,} vol.~15 no.~2 pp.~262--278, 2009.



    \bibitem{McElieceEtAl1987}
	R.J.\ McEliece et al.,
	``The capacity of the Hopfield associative memory,''
	\textit{IEEE Trans.\ on Inform.\ Theory,} vol.~IT-33 no.~4 pp.~461--482, July 1987.

	\bibitem{StorkeyValabregue1999}
    A.J.~Storkey and R.~Valabregue,
    ``The basins of attraction of a new Hopfield learning rule,''
    \emph{Neural Networks,} vol.~12, no.~6, pp.~869--876, July 1999.

    

    \bibitem{AlonSpencerBook}
    N.\ Alon and J.\ Spencer,
    \emph{The Probabilistic Method.}
    3rd Edition, Wiley \& Sons, 2008.
    
    \bibitem{LoeligerEtAl2007}
	H.-A.\ Loeliger et al.,
	``The factor graph approach to model-based signal processing,''
	\textit{Proc. of IEEE,} vol.~95 no.~6 pp.~1295--1322, June 2007.

	




\end{thebibliography}



\newpage
$\,$
\newpage

\appendix 

\subsection{Proof of Lemma~\ref{lem:Sigma}}
\label{appendixLemmaSigma}

First, note that
\begin{align}
    \ip{\vect{a}_{n-1}}{\tilde{\vect{a}}_{j-1}}
        &= \sum_{i=1}^L a_{i,n-1}\tilde{a}_{i,j-1} \\
        &= \sum_{i=1}^L a_{i,n-1}(a_{i,j-1} - p).
\end{align}
Now, for all $(i,j)\in\{1,\ldots,L\} \times \{1,\ldots,N\}$
we introduce the random variables $X_{i,j} \overset{\text{i.i.d.}}{\sim} \Ber(p)$.
Then, we have (see definition in (\ref{eq:Def_selln}))
\begin{align}
    S_{\ell,n} = \sum_{j\in\{1,\ldots,N\}\setminus\{n\}} \sum_{i=1}^L X_{\ell,j} X_{i,n-1} (X_{i,j-1} - p) \label{eq:DefSigmaWX}
\end{align}
and
\begin{align}
    \E&[S_{\ell,n}] \nonumber \\
        &\quad= \sum_{j\in\{1,\ldots,N\}\setminus\{n\}} \sum_{i=1}^L \E[X_{\ell,j} X_{i,n-1} (X_{i,j-1} - p)] \label{eq:X_linExpectation} \\
        &\quad= \sum_{j\in\{1,\ldots,N\}\setminus\{n\}} \sum_{i=1}^L \E[X_{\ell,j} X_{i,n-1}] \underbrace{\E[X_{i,j-1} - p]}_{\,\,\,\,\,=\,0} \label{eq:X_indep} \\
        &\quad= 0,
\end{align}
which proves (\ref{eq:lemmaS_A}).
In (\ref{eq:X_linExpectation}) we used linearity of expectation, 
and equation (\ref{eq:X_indep}) holds because the random variables
$X_{i,j-1}$ are independent of $X_{\ell,j} X_{i,n-1}$ for all
$(i,j)\in\{1,\ldots,L\}\times\{1,\ldots,N\}\setminus\{n\}$.

Using equation (\ref{eq:DefSigmaWX}), we have for all $t\in\R$
\begin{align}
    &\E\!\left[e^{tS_{\ell,n}}\right] \nonumber \\
        &\,= \E\!\left[ \prod_{j\in\{1,\ldots,N\}\setminus\{n\}} \prod_{i=1}^L e^{t X_{\ell,j} X_{i,n-1} (X_{i,j-1} - p)} \right] \\
        &\,= \E\!\left[ \prod_{j\in\{1,\ldots,N\}\setminus\{n\}} e^{t X_{\ell,j} X_{\ell,n-1} (X_{\ell,j-1} - p)} \right. \nonumber \\
        &\qquad\qquad \left. \cdot\prod_{i\in\{1,\ldots,L\}\setminus\{\ell\}} e^{t X_{\ell,j} X_{i,n-1} (X_{i,j-1} - p)} \right] \label{eq:MomentGen_A} \\
        &\,= \prod_{j=1}^N \sum_{x_{\ell,j}\in\{0,1\}} p(x_{\ell,j}) \prod_{i\in\{1,\ldots,L\}\setminus\{\ell\}} \sum_{x_{i,n-1}\in\{0,1\}} p(x_{i,n-1}) \nonumber \\
        &\qquad \cdot \E\!\left[ \prod_{j\in\{1,\ldots,N\}\setminus\{n\}} e^{t x_{\ell,j} x_{\ell,n-1} (x_{\ell,j-1} - p)} \right. \nonumber \\
        &\qquad\qquad \left. \cdot \prod_{i\in\{1,\ldots,L\}\setminus\{\ell\}} e^{t x_{\ell,j} x_{i,n-1} (X_{i,j-1} - p)} \right] \label{eq:MomentGen_B} \\
        &\,= \prod_{j=1}^N \sum_{x_{\ell,j}\in\{0,1\}} p(x_{\ell,j}) \prod_{i\in\{1,\ldots,L\}\setminus\{\ell\}} \sum_{x_{i,n-1}\in\{0,1\}} p(x_{i,n-1}) \nonumber \\
        &\qquad \cdot \prod_{j\in\{1,\ldots,N\}\setminus\{n\}} e^{t x_{\ell,j} x_{\ell,n-1} (x_{\ell,j-1} - p)} \nonumber \\
        &\qquad\qquad \cdot \prod_{i\in\{1,\ldots,L\}\setminus\{\ell\}} \E\!\left[ e^{t x_{\ell,j} x_{i,n-1} (X_{i,j-1} - p)} \right]. \label{eq:MomentGen_C}
\end{align}
In equation (\ref{eq:MomentGen_A}) we took the factor for $i=\ell$ out of the product on the right-hand side.
The step from (\ref{eq:MomentGen_A}) to (\ref{eq:MomentGen_B}) uses the law of total probability 
(conditioning on channel $\ell$ and time step $n-1$) together with the shorthand notation
\begin{align}
    p(x_{i,j}) := \Pr[X_{i,j}=x_{i,j}]\,,
\end{align}
and the fact that the remaining random variables $X_{i,j-1}$
are independent on the conditioning random variables $X_{\ell,j}$ and $X_{i,n-1}$
for $(i,j)\in\{1,\ldots,L\}\setminus\{\ell\}\times\{1,\ldots,N\}\setminus\{n\}$.
Finally, in (\ref{eq:MomentGen_C}) we used linearity of expectation
and independence of the remaining random variables.

Now, we upper bound
\begin{align}
    \E\!\left[ e^{t x_{\ell,j} x_{i,n-1} (X_{i,j-1} - p)} \right]
        &\leq e^{(t x_{\ell,j} x_{i,n-1})^2/8} \label{eq:UppBoundAlanSpencer_A} \\
        &\leq e^{t^2/8}, \label{eq:UppBoundAlanSpencer_B}
\end{align}
where (\ref{eq:UppBoundAlanSpencer_A}) follows from the inequality \cite[Lemma A.1.6]{AlonSpencerBook}
\begin{align}
        pe^{\lambda(1-p)} + (1-p)e^{-\lambda p} \leq e^{\lambda^2/8}, 
\end{align}
for all $0<p<1$, $\lambda\in\R$,
and inequality (\ref{eq:UppBoundAlanSpencer_B}) holds because $x_{\ell,j}, x_{i,n-1}\in\{0,1\}$.
Thus,
\begin{align}
    &\E\!\left[e^{tS_{\ell,n}}\right] \nonumber \\
        &\,\leq \prod_{j=1}^N \sum_{x_{\ell,j}\in\{0,1\}} p(x_{\ell,j}) \prod_{i\in\{1,\ldots,L\}\setminus\{\ell\}} \sum_{x_{i,n-1}\in\{0,1\}} p(x_{i,n-1}) \nonumber \\
        &\qquad \cdot \prod_{j\in\{1,\ldots,N\}\setminus\{n\}} e^{t x_{\ell,j} x_{\ell,n-1} (x_{\ell,j-1} - p)} \cdot e^{\frac{t^2}{8}(L-1)} \label{eq:uppBoundMG_A} \\
        &= e^{\frac{t^2}{8}(L-1)(N-1)} \prod_{j=1}^N \sum_{x_{\ell,j}\in\{0,1\}} p(x_{\ell,j}) \nonumber \\
        &\qquad \cdot \prod_{j\in\{1,\ldots,N\}\setminus\{n\}} e^{t x_{\ell,j} x_{\ell,n-1} (x_{\ell,j-1} - p)} \label{eq:uppBoundMG_B} \\
        &= e^{\frac{t^2}{8}(L-1)(N-1)} \sum_{x_{\ell,n-1}\in\{0,1\}} p(x_{\ell,n-1}) \nonumber \\
        &\qquad\quad \cdot\sum_{x_{\ell,n-2}\in\{0,1\}} p(x_{\ell,n-2}) e^{t x_{\ell,n-1} x_{\ell,n-1} (x_{\ell,n-2} - p)} \nonumber \\
        &\hspace{1.9cm} \vdots \nonumber \\
        &\qquad\quad \cdot \sum_{x_{\ell,1}\in\{0,1\}} p(x_{\ell,1}) e^{t x_{\ell,2} x_{\ell,n-1} (x_{\ell,1} - p)} \nonumber \\
        &\qquad\quad \cdot \sum_{x_{\ell,N}\in\{0,1\}} p(x_{\ell,N}) e^{t x_{\ell,1} x_{\ell,n-1} (x_{\ell,N} - p)} \nonumber \\
        &\hspace{1.9cm} \vdots \nonumber \\
        &\qquad\quad \cdot \underbrace{\sum_{x_{\ell,n}\in\{0,1\}} p(x_{\ell,n}) e^{t x_{\ell,n+1} x_{\ell,n-1} (x_{\ell,n} - p)}}_{\qquad = \,\E\left[e^{tx_{\ell,n+1}x_{\ell,n-1}(X_{\ell,n}-p)}\right] \,\,\leq\,\, e^{t^2/8}} \label{eq:Reordering} \\
        &\leq e^{\frac{t^2}{8}(L-1)(N-1)} e^{\frac{t^2}{8}(N-1)} \label{eq:uppBoundMG_C} \\
        &= e^{\frac{t^2}{8}L(N-1)},
\end{align}
which proves (\ref{eq:lemmaS_B}).
The step from (\ref{eq:MomentGen_C}) to (\ref{eq:uppBoundMG_A}) follows from (\ref{eq:UppBoundAlanSpencer_B}).
The factors in the product on the right-hand side of (\ref{eq:uppBoundMG_A}) do not depend on the variables
$\{x_{i,n-1}\}_{i\in\{1,\ldots,L\}\setminus\{\ell\}}$, thus
$\prod_{i\in\{1,\ldots,L\}\setminus\{\ell\}} \sum_{x_{i,n-1}\in\{0,1\}} p(x_{i,n-1})$
evaluates to one, which results in (\ref{eq:uppBoundMG_B}).
The reordering of terms in (\ref{eq:Reordering}) is based on the following observation.
For $j\in\{1,\ldots,N\}\setminus\{n\}$, let us define the factors
\begin{align}
    y_{\ell,j}(x_{\ell,j},x_{\ell,n-1},x_{\ell,j-1}) := e^{tx_{\ell,j} x_{\ell,n-1} (x_{\ell,j-1}-p)},
\end{align}
and then we analyze the dependencies in
\begin{align}
    &\prod_{j\in\{1,\ldots,N\}\setminus\{n\}} y_{\ell,j}
        = \prod_{j\in\{1,\ldots,N\}\setminus\{n\}} e^{tx_{\ell,j} x_{\ell,n-1} (x_{\ell,j-1}-p)} \label{eq:Reordering_C} \\
        &\quad= e^{tx_{\ell,1} x_{\ell,n-1} (x_{\ell,N}-p)} e^{tx_{\ell,2} x_{\ell,n-1} (x_{\ell,1}-p)} \cdots \nonumber \\
        &\qquad\quad \cdots e^{tx_{\ell,n-1} x_{\ell,n-1} (x_{\ell,n-2}-p)} e^{tx_{\ell,n+1} x_{\ell,n-1}(\textcolor{red}{\bm{x_{\ell,n}}}-p)} \cdots \nonumber \\
        &\qquad\quad \cdots e^{tx_{\ell,N} x_{\ell,n-1}(x_{\ell,N-1}-p)} \label{eq:Reordering_A} \\
        &\quad= e^{tx_{\ell,n+1} x_{\ell,n-1}(\textcolor{red}{\bm{x_{\ell,n}}}-p)} \cdots \nonumber \\
        &\qquad\quad \cdots e^{tx_{\ell,N} x_{\ell,n-1}(x_{\ell,N-1}-p)} e^{tx_{\ell,1} x_{\ell,n-1}(x_{\ell,N}-p)} \cdots \nonumber \\
        &\qquad\quad\cdots e^{tx_{\ell,n-1} x_{\ell,n-1}(x_{\ell,n-2}-p)}. \label{eq:Reordering_B}
\end{align}
We reordered the factors from $j=1,\ldots,n-1,n+1,\ldots,N$ in (\ref{eq:Reordering_A}) to
$j=n+1,\ldots,N,1,\ldots,n-1$ in (\ref{eq:Reordering_B}).
A factor graph (we use the same conventions as in \cite{LoeligerEtAl2007}) of the ``unfolded'' product
$y_{\ell,n+1}\cdots y_{\ell,N}y_{\ell,1}\cdots y_{\ell,n-1}$ is shown in Figure~\ref{fig:FG_product}.
The variable $\textcolor{red}{\bm{x_{\ell,n}}}$ only appears in the factor $y_{\ell,n+1}$,
which breaks up the dependence chain.
Thus, closing boxes via upper bounding the corresponding expectation using (\ref{eq:UppBoundAlanSpencer_B}),
from left to right in Figure~\ref{fig:FG_product} yields inequality (\ref{eq:uppBoundMG_C}).
\hfill $\blacksquare$


\subsection{Chernoff Bound on the Lower Tail of Binomial Distribution}
\label{appendixChBoundBinomial}
  
\begin{lem} \label{lem:ChBoundBinomial}
Let $X \sim \Bin(L,p)$, then for $0< \delta < 1$
\begin{align}
    \Pr\!\left[X \leq (1-\delta)Lp\right] \leq e^{-D_{\text{KL}}\left((1-\delta)p\parallel p\right) L}, \label{eq:ChBoundBin}
\end{align}
with Kullback--Leibler divergence (or relative entropy)
\begin{align}
    D_{\text{KL}}(p_1\!\!\parallel\!p_2) := p_1 \ln\!\left(\frac{p_1}{p_2}\right) + (1-p_1) \ln\!\left(\frac{1-p_1}{1-p_2}\right)\!,
\end{align}
for $0<p_1,p_2<1$. \hfill $\square$
\end{lem}

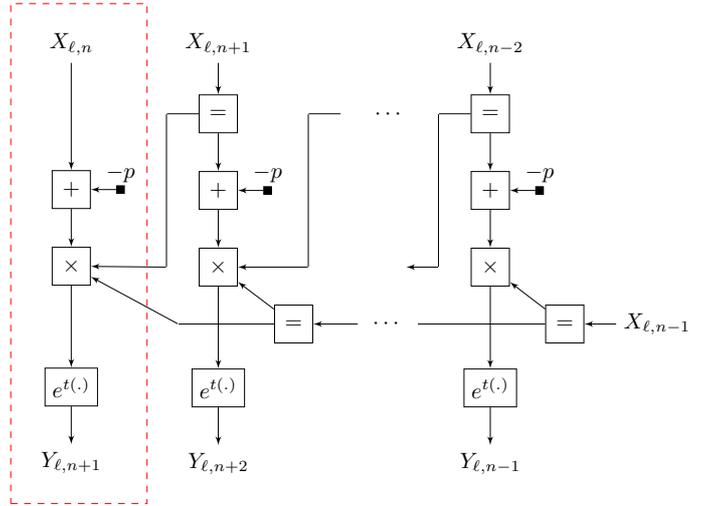
\begin{figure}
    \resizebox{0.51\textwidth}{!}{%
        \centering
\begin{tikzpicture}
[node distance=19mm,auto,>=latex',
  box/.style={draw, minimum size=0.6cm},
  short/.style={node distance=14mm},
  dash/.style={draw=red,dashed, minimum size=0.5cm},
  fixi/.style={draw,fill=black, inner sep=0.06cm},
  scale=1,transform shape]

\node (X_n) {$X_{\ell,n}$};
\node[box,below =1.7cm of X_n] (plus_n) {$+$} edge[<-] (X_n);
\node[fixi,right =0.4cm of plus_n] (shift_n) {} edge[->] (plus_n);
\node[above] at (shift_n) {$-p$};
\node[box,below =0.6cm of plus_n] (mul_n) {$\times$} edge[<-] (plus_n);
\node[box,below =1.3cm of mul_n] (exp_n) {$e^{t(.)}$} edge[<-] (mul_n);
\node[below =0.6cm of exp_n] (Y_{n+1}) {$Y_{\ell,n+1}$} edge[<-] (exp_n);
\node[right =1.2cm of X_n] (X_{n+1}) {$X_{\ell,n+1}$};
\node[box,below =0.5cm of X_{n+1}] (eq_{n+1}) {$=$} edge[<-] (X_{n+1});
\node[box,below =0.6cm of eq_{n+1}] (plus_{n+1}) {$+$} edge[<-] (eq_{n+1});
\node[fixi,right =0.4cm of plus_{n+1}] (shift_{n+1}) {} edge[->] (plus_{n+1});
\node[above] at (shift_{n+1}) {$-p$};
\node[box,below =0.6cm of plus_{n+1}] (mul_{n+1}) {$\times$} edge[<-] (plus_{n+1});
\node[box,below =1.3cm of mul_{n+1}] (exp_{n+1}) {$e^{t(.)}$} edge[<-] (mul_{n+1});
\node[below =0.6cm of exp_{n+1}] (Y_{n+2}) {$Y_{\ell,n+2}$} edge[<-] (exp_{n+1});
\node[inner sep=0pt,minimum size=0pt, left =0.5cm of eq_{n+1}] (con_1) {};
\draw (con_1) -- (eq_{n+1});
\node[inner sep=0pt,minimum size=0pt, left =0.5cm of mul_{n+1}] (con_2) {} edge[->] (mul_n);
\draw (con_1) -- (con_2);
\node[inner sep=0pt,minimum size=0pt, right =1.1cm of mul_{n+1}] (con_3) {} edge[->] (mul_{n+1});
\node[inner sep=0pt,minimum size=0pt, right =1.1cm of eq_{n+1}] (con_4) {};
\draw (con_3) -- (con_4);
\node[inner sep=0pt,minimum size=0pt, right =0.5cm of con_4] (con_5) {};
\draw (con_4) -- (con_5);
\node[right =0.4cm of con_5] (dots) {$\cdots$};
\node[inner sep=0pt,minimum size=0pt, right =0.4cm of dots] (con_6) {};
\node[box,right =0.5cm of con_6] (eq_{n-1}) {$=$};
\node[above =0.5cm of eq_{n-1}] (X_{n-2}) {$X_{\ell,n-2}$} edge[->] (eq_{n-1});
\node[box,below =0.6cm of eq_{n-1}] (plus_{n-1}) {$+$} edge[<-] (eq_{n-1});
\node[fixi,right =0.4cm of plus_{n-1}] (shift_{n-1}) {} edge[->] (plus_{n-1});
\node[above] at (shift_{n-1}) {$-p$};
\node[box,below =0.6cm of plus_{n-1}] (mul_{n-1}) {$\times$} edge[<-] (plus_{n-1});
\node[box,below =1.3cm of mul_{n-1}] (exp_{n-1}) {$e^{t(.)}$} edge[<-] (mul_{n-1});
\node[below =0.6cm of exp_{n-1}] (Y_{n-1}) {$Y_{\ell,n-1}$} edge[<-] (exp_{n-1});
\node[inner sep=0pt,minimum size=0pt, left =0.5cm of mul_{n-1}] (con_7) {};
\draw (eq_{n-1}) -- (con_6);
\draw (con_6) -- (con_7);
\node[inner sep=0pt,minimum size=0pt, left =0.5cm of con_7] (con_8) {} edge[<-] (con_7);
\node[right =0.4cm of X_{n-2}] (aux) {};
\node[box,below =4cm of aux] (eq_n) {$=$} edge[->] (mul_{n-1});
\node[right =0.5cm of eq_n] (X_{n-1}) {$X_{\ell,n-1}$} edge[->] (eq_n);
\node[inner sep=0pt,minimum size=0pt, left =2cm of eq_n] (con_9) {};
\draw (eq_n) -- (con_9);
\node[left =0.1cm of con_9] (dots_1) {$\cdots$};
\node[inner sep=0pt,minimum size=0pt, left =0.1cm of dots_1] (con_11) {};
\node[box,left =0.7cm of con_11] (eq_X) {$=$} edge[<-] (con_11);
\draw[->] (eq_X) -- (mul_{n+1});
\node[inner sep=0pt,minimum size=0pt, left =1.5cm of eq_X] (con_10) {} edge[-] (eq_X);
\draw[->] (con_10) -- (mul_n);
\node[box,dashed,red,inner sep=3.5mm,fit=(X_n)(shift_n)(Y_{n+1})] (closedBox) {};
\node[inner sep=0pt,minimum size=0pt, above =1.9cm of con_1] (label) {$\color{red} \E[e^{tx_{\ell,n+1}x_{\ell,n-1}(X_{\ell,n}-p)}] \leq e^{t^2/8}$};
\end{tikzpicture}
    }%
    \caption{Factor graph of (\ref{eq:Reordering_C}).}
    \label{fig:FG_product}
\end{figure}

This is undoubtedly well known, but for the convenience of the reader, we give a proof.

\begin{proof}
First we note that for every $t<0$
\begin{align}
    \Pr\!\left[X\leq (1-\delta)Lp\right] = \Pr\!\left[e^{tX} \geq e^{t(1-\delta)Lp}\right]\!, \label{eq:ChBoundBin_A}
\end{align}
and then applying Markov's inequality to the right-hand side of (\ref{eq:ChBoundBin_A}) yields
\begin{align}
    \Pr\!\left[e^{tX} \geq e^{t(1-\delta)Lp}\right] &\leq \min_{t<0} \frac{\E\!\left[e^{tX}\right]}{e^{t(1-\delta)Lp}} \\
                                                     &= \min_{t<0} \frac{\left(1-p+pe^t\right)^L}{e^{t(1-\delta)Lp}} \\
                                                     &= \min_{t<0} f(t)^L
\end{align}
with
\begin{align}
    f(t) &:= \frac{1-p+pe^t}{e^{t(1-\delta)p}} \\
         &= (1-p)e^{-t(1-\delta)p} + pe^{t\left(1-(1-\delta)p\right)}.
\end{align}
Setting the derivative of $f$
\begin{align}
    f'(t) &= -(1-\delta)p(1-p)e^{-t(1-\delta)p} \nonumber \\
          &\qquad + \left(1-(1-\delta)p\right)\!pe^{t\left(1-(1-\delta)p\right)}
\end{align}
equal to zero yields an unique $t_{\min}$ with
\begin{align}
    e^{t_{\min}} = \frac{(1-\delta)(1-p)}{1-(1-\delta)p} < 1 \quad \implies t_{\min}<0,
\end{align}
and thus
\begin{align}
    \min_{t<0} f(t) &= f(t_{\min}) \\
                    &= e^{-t_{\min}(1-\delta)p}\left( 1-p + pe^{t_{\min}} \right) \\
                    &= \left(\frac{(1-\delta)(1-p)}{1-(1-\delta)p}\right)^{-(1-\delta)p} \nonumber \\
                    &\qquad \cdot\left( 1-p + p\frac{(1-\delta)(1-p)}{1-(1-\delta)p} \right) \\
                    &= (1-\delta)^{-(1-\delta)p} \nonumber \\
                    &\qquad \cdot\left(\frac{1-p}{1-(1-\delta)p}\right)^{-(1-\delta)p} \frac{1-p}{1-(1-\delta)p} \\
                    &= (1-\delta)^{-(1-\delta)p} \nonumber \\
                    &\qquad \cdot\left(\frac{1-(1-\delta)p}{1-p}\right)^{-\left(1-(1-\delta)p\right)} \\     
                    &= e^{-(1-\delta)p \ln(1-\delta)} e^{-\left(1-(1-\delta)p\right) \ln\left(\frac{1-(1-\delta)p}{1-p}\right)} \\
                    &= e^{-D_{\text{KL}}\left((1-\delta)p\parallel p\right)}. \label{eq:ChBoundBin_B}
\end{align}
The right-hand side of (\ref{eq:ChBoundBin_B}) is indeed the global minimum of $f$,
because $f(t_{\min}) < f(0)$ and $f(t_{\min}) < \lim_{t\to-\infty} f(t)$.
\end{proof}

\end{document}